\documentclass{llncs}
\usepackage{graphicx}
\usepackage[caption=false]{subfig} 
\usepackage{amssymb}
\usepackage{mathtools}
\usepackage{bm}
\usepackage{latexsym}
\usepackage{tikz}
\usepackage{lipsum}
\usepackage{microtype}

\newcommand{\agentA}{\ensuremath{\mathrm{A}}}
\newcommand{\agentAp}{\ensuremath{\mathrm{A}'}}
\newcommand{\agentB}{\ensuremath{\mathrm{B}}}
\newcommand{\agentC}{\ensuremath{\mathrm{C}}}
\newcommand{\agentD}{\ensuremath{\mathrm{D}}}

\newcommand{\hold}[1]{#1_\circlearrowright}
\newcommand{\entry}[1]{\prescript{}{\to}{#1}}
\newcommand{\exit}[1]{#1_\to}
\DeclareMathOperator{\pick}{pick}
\DeclareMathOperator{\drop}{drop}

\DeclarePairedDelimiter\floor{\lfloor}{\rfloor}

\title{On deterministic, constant memory triangular searches on the integer lattice}
\author{J. Alfredo Cruz-Carl\'on\inst{1}}
\institute{Universidad Nacional Aut\'onoma de M\'exico\\jucruz@ciencias.unam.mx}
\date{}
\begin{document}
\maketitle
\begin{abstract}
Recently it has been shown that four constant memory, deterministic agents are able to discover the integer lattice if only local, constant-size communication is allowed. Moreover, if the agents' choices are determined with the help of a fair coin, it has been shown that three are necessary and sufficient to discover the integer lattice. In this paper, we show that three deterministic agents cannot find the integer lattice and sketch a possible characterization for \emph{one explorer, three beacons} type of exploration algorithms. 

\end{abstract}

\section{Introduction}
In \cite{feinerman2012collaborative} Feinerman, et al. introduce the \emph{Ants Nearby Treasure Search} (ANTS) problem, which is a generalization of the cow-path problem
\cite{baezayates1993searching}. While this model (and a randomized version) has been studied, for example, Cohen, et al. in \cite{Cohen:2017:EIS:3039686.3039700} recently proved that two randomized agents can not find the treasure in the integer lattice, it remains (to the best of our knowledge) an open problem (\cite{emek2015many}) whether three deterministic agents can find the treasure. In this paper we show that they can not.

Our approach is based on two insights. The first is, communication is key; as we will see, if an agent gets lost, the region of the plane it can explore on its own is limited. The second (insight) is actually a question; if communication is key and we know two agents already met at some cell (for example we assume all agents start inside the same cell) what must happen so they meet again? To answer this question, we will explore the role of the automata, the scheduling policy and the communication protocol to determine the conditions every meeting must have. Following we will use these conditions to characterize the shape of the area of the integer plane one, two and three agents can effectively discover and finally we sketch a possible characterization of \emph{one explorer, three beacons} exploration algorithms. 

The paper is structured as follows: in Section~\ref{sec:model} we define our model, following in Section~\ref{sec:ep} we develop or necessary tools and finally present our main result as well as a consequence. Finally in Section~\ref{conc} we sketch a characterization of \emph{one explorer, three beacons} exploration algorithms.


\section{Model}
\label{sec:model}



Let \agentA{} be an agent, formally it is a 3-tuple $\Pi_\agentA = \langle Q, s_0^\agentA, \delta\rangle$ where $Q$ is the finite \emph{set of states}, $s_0^\agentA$ is its \emph{start state} and $\delta: Q\times 2^Q \rightarrow Q\times H$ is the \emph{transition function} where $H = \{N,S,W,E,P\}$\footnote{Since we focus on deterministic protocols only, the codomain of $\delta$ is $Q\times H$ instead of  $2^{Q\times H}$ adopted in \cite{emek2015many}  }. The letters $N,S,E,W$ correspond to the four cardinal points and $P$ stands for \emph{stay put}, i.e. the agent stays inside the current cell. An agent \agentA{} may communicate with another agent \agentB{} if both are inside the same cell. Moreover this exchange is of constant size, they only sense the other's current state (this communication is modeled by the power set of $Q$ in the $\delta$'s domain). If two agents are inside the same cell we say they are \emph{neighbors}. Throughout the paper we use the letters \agentA{}, \agentB{}, and \agentC{} to denote agents. 


An agent may be in three states, active, paused or halted. An agent is \emph{active} at time $t$ if it performs a $\delta$-transition, is \emph{paused} (at time $t$) when it does not and it is \emph{halted} if it is in the final state. To determine when an agent becomes active (inactive) we assume the existence of a discrete global clock, an \emph{activation function} $t$ for each agent (e.g. $t_\agentA$ for \agentA{}) and a \emph{scheduling policy} $\psi$.
When an agent performs a $\delta$-transition we assume the change of its state and its movement happen at the same time; we denote the time of the $i$-th $\delta$-transition by $t_\agentA(i)$.

We call a \emph{protocol} to a set of agents and their automaton. A \emph{schedule} is an adversary controlled activation and pause actions of the automaton in a protocol. A protocol is \emph{effective} or \emph{schedule resistant} if it can find the treasure for any given schedule.

Because all agents have the same states set $Q$ and the same transition function $\delta$, then, all agents are indistinguishable. Let $Q_a \subseteq Q$ be the set of states that are reachable from $s^0_a$. We say two agents $a,b$ are \emph{disjoint} if $Q_a \cap Q_b = \emptyset$. In this paper we focus on disjoint agents and their interaction. 


A states sequence $L = \langle q_l, \ldots, q_i,q_j,\ldots,q_k\rangle$ is a \emph{chain} if for any two consecutive states in $L$, $q_i,q_j$, $\delta(q_i,a) = (q_j,b)$ for any $a \subset Q, b \in H$.  If $\delta(q_k,a)=(q_l,b)$ we say $L$ is a \emph{cycle} and that $q_l$ and $q_k$ are consecutive. Given two chains $J,L$ we denote by $J \odot L$ the chain resulting from connecting the last state in $J$ ($j_k$) with the first state of $L$ ($l_1$), such connection is made by defining $\delta(j_k,a) = (l_1,b), a \subset Q, b \in H$. We denote the length of $L$ by $|L|$. An agent \emph{follows} $L$ if its automaton starts at $q_l$ and ends at $q_k$. In the case $L$ is a cycle we also say an agent follows $L^k$ meaning it follows $L$ $k$ times.
For two chains $L, J, L \leadsto J$ denotes there is a chain, possibly empty, connecting $L$ and $J$.

Let $C = \langle \ldots, q_k,q_i, q_j, \ldots \rangle$ be a cycle and $q$ a state such that $\delta(q_i,a) = (q,b), q \neq q_j$. 
We call the ordered pair $(a,q_i)$ an \emph{exit condition} of $C$ at $q_i$. 
Analogously we call $(a,q_i)$ an \emph{entry condition} if $\delta(q,a) = (q_i,b), q \neq q_k$. 
For any two consecutive states in $C$ e.g. $q_k,q_i$ such that $\delta(q_k,a) = (q_i,b)$ we call $a$ a \emph{holding condition}.
We denote by $\entry{C}$ to the set of all entry conditions, $\exit{C}$ the set of all exit conditions and $\hold{C}$ the set of all holding conditions.
Intuitively $\entry{C}, \hold{C}$ and $\exit{C}$ describe the conditions on the exploration area for the agent to start, continue or stop following $C$.



\section{Exploration power}
\label{sec:ep}
The goal of this section is to determine the maximal area of the integer plane three agents can effectively discover or in other words their exploration power. Formally the \emph{exploration power} of $k$ agents is the shape of the maximal area of the integer plane for which there exists an effective protocol. It relies on two aspects, how the agents move and interact. We say \agentA{} and \agentB{} \emph{interact} if they meet in an unbounded number of cells. As we will see (Subsection~\ref{lab:OneAnt}) all the agents move in a similar fashion when they are alone; either they continue blindly until they meet another agent or they move a prefixed number of cells and stop waiting for an agent to arrive at the same cell. On the other hand, to characterize the interaction between two agents we must first determine the conditions under which they meet in a cell (Subsection~\ref{lab:twoAntsSec}). Finally, we will characterize the exploration power of three agents (Subsection~\ref{lab:three_ep}). 

\subsection{One ant}
\label{lab:OneAnt}
In the case of one ant, \agentA{}, there are no meetings, hence we focus on the structure of the automata when it explores an unbounded region of the plane. 
Since \agentA{} has a finite number of states, it can only move arbitrary long distances by a cycle $C = \langle c_0,\ldots,c_k\rangle$. 
The overall structure of \agentA{} must be $I \leadsto C \leadsto  h$ where $I$ is a possibly empty chain (i.e. $s_0^\agentA$ belongs to $C$) and $h$ is the halting state. It must be that $\entry{C} = \{(\emptyset,c_j)\}, 0 \leq j \leq k, \exit{C}=\{(\tau, h)\},\hold{C} =\{\emptyset\}$. We call any cycle $C$ such that $\hold{C} =\{\emptyset\}$ an \emph{exploration cycle}.

\begin{figure}
\centering
\subfloat[]{
	\includegraphics[width=0.3\textwidth]{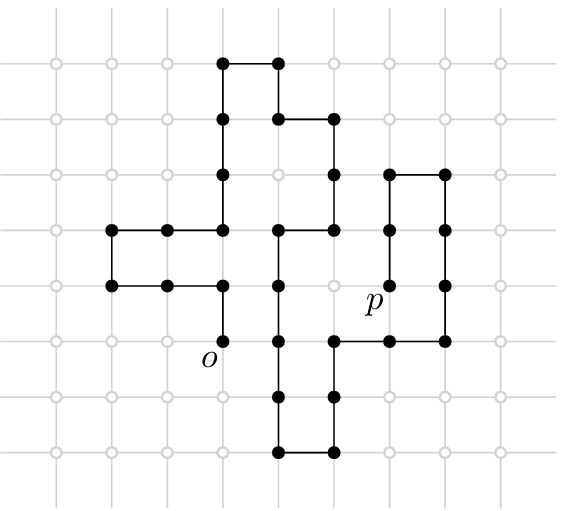}
}
\subfloat[]{
	\includegraphics[width=0.22\textwidth]{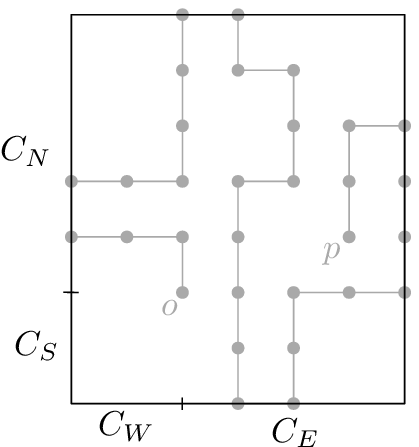}
}
\label{OneAgent1}
\caption{The bounding rectangle (right) of the path from $o$ to $p$ (left).}
\end{figure}

Suppose \agentA{} follows $C$ from a start cell $o$ and when it performs the transition $\delta(c_k, \emptyset) = (c_0, h)$ it stops at cell $p$. There are momments where \agentA{} reaches furthest south, west, north and east. Hence, there exist a rectangle $R_C$ that bounds \agentA{} movements with sides of size $C_W + C_E$ and $C_S + C_N$ where $C_h$ is the maximum number of cells in the heading $h$ from $o$ (see Figure~\ref{OneAgent1}).

Consider an agent \agentAp{ }with overall structure $ I \leadsto C'\leadsto h$ such that it explores the whole $R_C$ area when \agentAp{ }follows $C'$ from $o$ to $p$. One possible structure for $C'$ is to first go to the south-west corner of $R_C$, $S^{C_S}W^{C_W}$, that is move $C_S$ cells to the south and $C_W$ cells to the west; then explore each column of $R_C$, $(N^{C_S+C_N}S^{C_S+C_N}E)^{C_W+C_E}W$. Move from the south-east corner back to $o$, $W^{C_E}N^{C_S}$ and then move from $o$ to $p$, $H^{O_H}H'^{O_{H'}}$; where $H$ and $H'$ are the headings to get from $o$ to $p$ and $O_H, O_{H'}$ their respective number of cells.  Putting all together, a possible movement pattern for $C'$ is: $S^{C_S}W^{C_W}(N^{C_S+C_N}S^{C_S+C_N}E)^{C_W+C_E}W^{C_E+1}N^{C_S}H^{O_H}H'^{O_{H'}}$. To ensure the agent halts when it finds the treasure, each state in $C'$ has a transition to $h$ labeled with $\tau$. We have the following lemma:

\begin{lemma}
\agentA{ }movements are bounded by \agentAp{}.
\end{lemma}

\begin{figure}[h]
	\centering
	\subfloat[A half-band, the agent explores a black rectangle and then a gray rectangle.]{
		\includegraphics[width=0.33\textwidth]{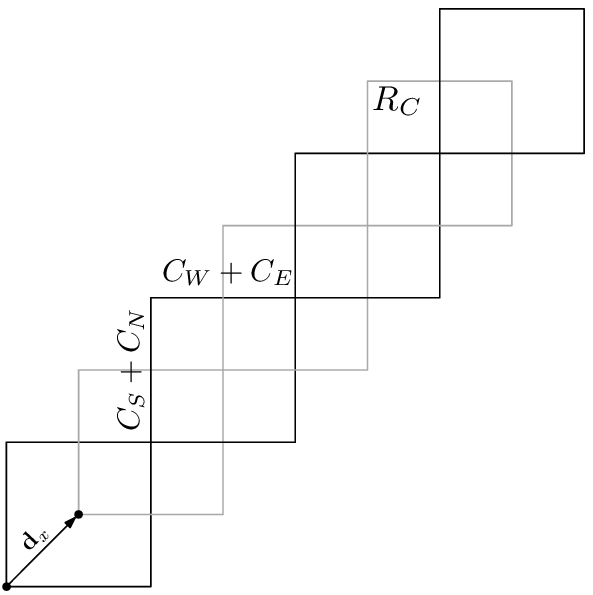}
	}
	\,\,\,\,\,\,
	\subfloat[The exploration power of three agents in the absence of communication]{
		\includegraphics[width=0.4\textwidth]{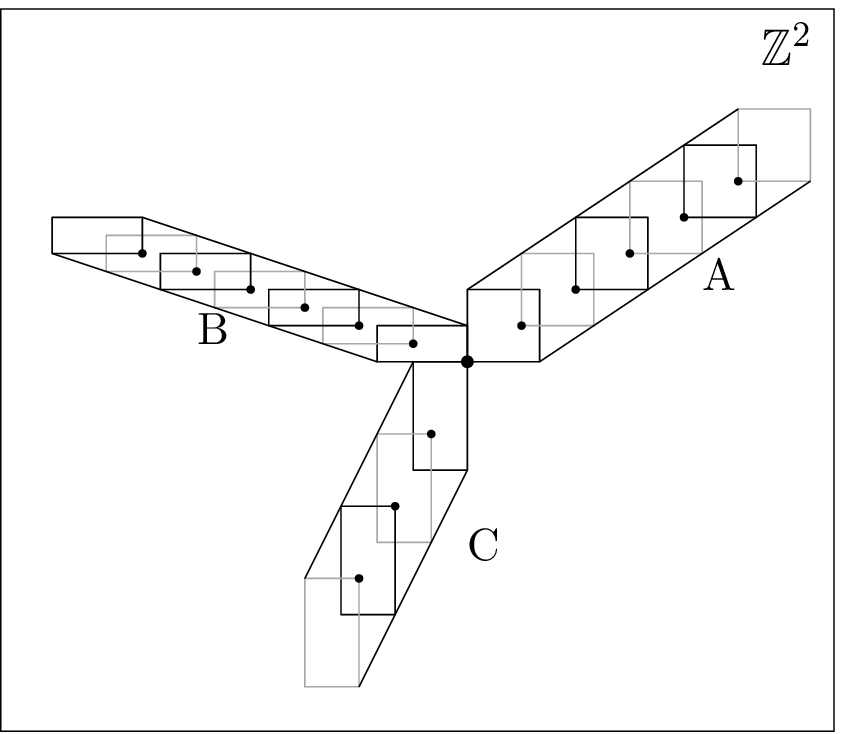}
	}
	\caption{A half-band and illustration of Corollary~\ref{3exp_pow_1}}
	\label{fig:half-band}
\end{figure}

We call the pattern described by \agentAp{ }(basically a set of overlapping rectangles, see Figure~\ref{fig:half-band}) a \emph{half-band} with \emph{width} $C_W+C_E$, \emph{height} $C_S+C_N$ and \emph{continuation} $O_H+ O_{H'}$. If $O_H = O_{H'} = 0$ we call the cycle $C$ \emph{stationary} (the agent explores the same rectangle over and over) we call it \emph{no-stationary} otherwise. To ease our presentation, we define the direction of a half-band as the difference vector between the coordinates of the end cell $p$ minus the coordinates of the start cell $o$ and denote it by $\vec{d}_x$ where $x$ is an agent.

\begin{corollary}
\label{exp_pow_1ant}
The exploration power of one ant is a half-band.
\end{corollary}

We can also see that when an agent explores by means of an exploration cycle, then it will explore at most a half-band.

\begin{corollary}
\label{1exp_pow}
If \agentA{ }follows a cycle $C$ such that $\hold{C}= \{\emptyset\}$, then it explores at most a half-band.
\end{corollary}

In the remaining of the paper, we assume that all agents who follow an exploration cycle will explore half-bands.

\begin{corollary}
\label{3exp_pow_1}
In the absence of communication the exploration power of three agents are three independent half-bands.
\end{corollary}

\subsection{Two ants}
\label{lab:twoAntsSec}
In this subsection we determine the exploration power for two agents. To do so, we first introduce the concepts of waiting and meeting from the automaton's perspective. Following we characterize the conditions and properties of the different types of meetings to finally determine their exploration power.




An agent may be in one of three different behaviors, \agentA{} may be waiting for another agent inside a cell, moving and expecting to meet another agent or moving oblivious to any other agent. 

Formally \agentA{} \emph{waits for} \agentB{} if \agentA{} is inside a cycle $\omega_{\agentB}$ such that: $\hold{{\omega_\agentB}} = \Sigma \setminus S$ for some set $S \subset 2^{Q_\agentB}$ such that for all $s \in S$ exists $q \in Q_\agentA$ such that $(s, q) \in \exit{{\omega_\agentB}}$. We say \agentA{} \emph{expects a meeting with} \agentB{} \emph{at state} $q$ if if there is a transition from $q$ to another state with a label $l \in 2^{Q_\agentB}, l \neq \emptyset$. Finally \agentA{} \emph{moves obliviously} if \agentA{} follows a chain $\Theta_\agentA$ such that if $p$ and $q$ are adjacent states, then $\delta(p,s) = (q,a)$ for all $s \in 2^Q$. By the definition of an agent moving obliviously, it follows that if $\Theta_\agentA$ is a cycle, the agent will explore a half-band and only stop if it finds the treasure. Hence we will assume $\Theta_\agentA$ is not a cycle.

\emph{Arranged} meetings are when \agentA{ }and \agentB{ }move obliviously a number of cells and then both enter in a waiting state to wait for the other's arrival.
\emph{Discovered} meetings are when \agentA{ }and \agentB{ }follow no-stationary cycles and they meet at a cell \emph{by chance}.
Finally, a \emph{scheduled} meeting is when \agentA{ }moves following a no-stationary cycle while \agentB{ }moves obliviously a number of cells and then waits for \agentA{}. 

Arranged meetings occurrence depends solely on the design of the protocol. Agents \agentA{ }and \agentB{ }explore a prefix set of cells before meeting and both will enter in a waiting state at the meeting cell, preventing one agent overtaking the other. Arranged meetings require knowledge of the agents relative position to each other when the protocol is being designed; if \agentA{} and \agentB{} only move in arranged meetings, 
both agents basically travel along (possibly) different paths to the next meeting cell.

Discovered meetings occurrence depends mostly on the scheduler. If the half-bands are of equal size and \emph{in opposite directions} then both agents will meet independently of the schedule, however the schedule determines the exact position of the meeting. For example, if \agentA{ }is at cell $(i,j)$ and only moves to the east and \agentB{ }is at cell $(l,j)$ with $i < l$ and only moves to the west, it is schedule dependent where exactly in the $[i,l]$ interval the agents will meet.
On the other hand if the agents are not in opposite directions, then there exists a schedule that can make the agents fail their meeting, its enough for a schedule to stop an agent until the other crosses the meeting area (in the case the half-bands only intersect in a finite area) or always keeping one agent before the other in the case both agents travel in the same direction.


Scheduled meetings have the best of arranged and discovered meetings, they have the \emph{fixed} meeting points of arranged meetings and the unbounded exploration power of discovered meetings, moreover, they only rely on careful protocol design to be schedule resistant. We call the agent that follows the no-stationary cycle the \emph{unbounded} agent and the agent that moves a prearranged set of cells and then waits for the other, the \emph{bounded} agent. We now explore the structure of a scheduled meeting.

In the remaining of this subsection, let \agentA{ }be the unbounded agent following an exploration cycle $\sigma$, and \agentB{ }the bounded one, $\mathbb{P}$ the scheduled meeting protocol from cell $c_1$ to cell $c_k$, $c_1 \neq c_k$ and let $T_\agentB = (c_1,\ldots,c_k)$ be the trajectory of cells \agentB{ }explores while executing $\mathbb{P}$. 

A design requirement for the unbounded ant is that it must expect to meet the bounded ant at any state of its exploration cycle or it might \emph{over take it}. 
Without loss of generality and for simplicity,  we will assume that, under $\mathbb{P}$, if \agentA{ }and \agentB{ }meet at any cell (including the start cell), \agentA{ }will always wait for \agentB{ }to exit the cell.

\begin{lemma}
\label{last-known-pos}
Let \agentA{ }and \agentB{ }be inside the same cell $c_i$, $i < k$. When \agentA{ }resumes its execution of $\mathbb{P}$ then $c_{i+1}$ will be visited by \agentA{ }at least once.

\end{lemma}

\begin{proof}
If \agentA{ }does not visit $c_{i+1}$ at some point under $\mathbb{P}$, then an adversary making \agentB{ }wait long enough at $c_{i+1}$ for \agentA{ }to visit and leave $c_k$ for the last time and then resuming \agentB{ }execution, will succeed in disrupting the meeting. \qed
\end{proof}

\begin{corollary}
\label{ta_tbCoro}
Let $T_\agentA$ ($T_\agentB$) be the cells that \agentA{} (\agentB{}) visits while executing $\mathbb{P}$. Then $T_\agentB \subseteq T_\agentA$. Moreover for each $c_i \in T_\agentB$ the last time \agentA{ }visits $c_i$ is after \agentB{ }visited it for the last time.
\end{corollary}
\begin{proof}
By Lemma~\ref{last-known-pos} \agentA{ }will \emph{follow} \agentB{ }. Since any schedule does not change $\mathbb{P}$, then,  $T_\agentB \subseteq T_\agentA$.
Suppose there exists a $c_i$ such that \agentA{ }visits, for the last time, before \agentB{ }visits it for the last time. Hence any adversary that pauses $B$ at $c_i$ long enough for \agentA{ }to reach and leave $c_k$ for the last time will succeed in preventing the meeting. \qed
\end{proof}

\begin{lemma}
\label{2antsSchedMeet}
Any protocol for only two agents \agentA{ }and \agentB{ }that uses scheduled meetings is bounded by a protocol using only arranged meetings.
\end{lemma}
\begin{proof}
Without loss of generality we assume \agentA{ }and \agentB{ }start inside the same cell $c_0$. After \agentB{ }leaves $c_0$ and arrives at $c_1$, \agentA{ }may start to move. We argue that $c_1$ is among the cells explored in the first iteration of \agentA{}'s exploring cycle $\sigma$. Wlg, suppose $c_1$ is the east neighbor of $c_0$. If $c_1$ is not inside the first rectangle of the half-band $B_\sigma$, then $B_\sigma$ is to the west of $c_1$ because it must include $c_0$, hence the continuation point is to the west of $c_1$. Therefore \agentA{}'s half-band exploration is to the west of $c_1$. Since there is only two agents, \agentA{ }would continue exploring to the west of $c_1$ and it will never return to $c_1$, hence \agentA{ }and \agentB{ }will not meet. Therefore $c_1$ is among the cells explored in the first iteration of \agentA{}. 

The previous observation and Corollary~\ref{ta_tbCoro} implies that when there is only two agents, \agentA{ }follows \agentB{} but \agentA{ }may explore a finite number of \emph{extra cells}. The same behaviour can be achieved by \emph{gluing} together the number of rectangles $B_\sigma$ it takes $\sigma$ to reach $c_k$. Suppose $l$ rectangles are required for \agentA{ }to reach $c_k$ from $c_0$, then  we can construct \agentAp{ }such that its automaton has the overall structure \agentAp{}:$I \leadsto \sigma^l\omega_\mathcal{B}\leadsto h$ where the exit transition of $\omega_\mathcal{B}$ leads to the first state of $\sigma^l$. By the construction of \agentAp{ }, the agents \agentAp{ }and \agentB{ } explore using arranged meetings only. \qed
\end{proof}


\begin{lemma}
\label{exp_arranged}
The exploration power of two agents with arranged meetings is a half-band.
\end{lemma}
\begin{proof}
Given a protocol $\mathbb{M}$ for \agentA{ }and \agentB{ }that uses arranged meetings, we can construct a protocol $\mathbb{O}$ such that it uses only one agent, \agentC{},  and bounds $\mathbb{M}$. Recall that arranged meetings require the agents to move a prefixed amount of cells and then wait. For the sake of our construction, we assume the agents start at the same cell $c_0$. Let the agents have the overall structure $I_x \leadsto E_x \leadsto h_x$ where $x \in \{\agentA, \agentB\}$; $I_x$ is a trajectory from the initial state to the first state of the \emph{exploring} part of the automaton $E_x$, and $h_x$ is the halting state. Please note that $E_x$ is not composed of a single cycle, because at least a waiting cycle must exist in $E_x$. However the overall structure of $E_x$ must be $E_x: M_{x_1}\omega_{-x}\ldots M_{x_i}\omega_{-x}\ldots M_{x_n}\omega_{-x}M'_{x}$ where $M_{x_i}, i \geq 1$ are the prefixed movements for the next meeting, $\omega_{-x}$ is the waiting cycle and $M'_{x}$ is a chain (possibly empty) that leads back to $M_{x_1}$. Please note that at the end of each $M_{x_i}$ \agentA{ }and \agentB{ }meet.

Since we only have those two agents we can construct the chain $M_x = \bigodot_{i=1}^n M_{x_i}\odot M'_x$ where $\bigodot_{i=1}^n$ is a short hand of $M_{x_1}\odot\ldots \odot M_{x_i}\odot\ldots\odot M_{x_n}$. The chain $M_\agentC = M_\agentA\odot M_\agentB$ covers the area explored by \agentA{ }and \agentB{ }respectively and ends at the last meeting point. Therefore the area covered by \agentA{ }and \agentB{ }is covered by \agentC{ }alone. By Lemma~\ref{exp_pow_1ant} $\mathbb{M}$ is bounded $\mathbb{O}$. However if the agents explore in two different directions, the area of the integer plane they would cover is greater than just a half-band. \qed
\end{proof}

\begin{corollary}
\label{exp_pow_2ant}
The exploration power of two agents is two independent half-bands.
\end{corollary}

\begin{corollary}
\label{exp_pow_3ant_2inter}
Given three agents, if only two interact, then the exploration power is two independent half-bands.
\end{corollary}




\subsection{Three ants}
\label{lab:three_ep}
In subsection~\ref{lab:twoAntsSec} we determined the exploration power of two agents, in particular, we determined (Lemma~\ref{2antsSchedMeet}) that two agents using scheduled meetings have the same exploration power than two agents using arranged meetings. Please recall the equivalence emerged from the fact that the unbounded agent was only able to explore a fixed number of cells before meeting back with the bounded agent. This limitation is a consequence of the inability of the unbounded agent to \emph{turn around} at an arbitrary point in time (Lemma~\ref{1exp_pow}). Moreover, in Corollary~\ref{3exp_pow_1} we characterized how much of the integer plane three agents can cover when they do not interact pairwise. In this subsection we will show that the two remaining interaction scenarios (one agent interacts with two and pairwise interaction) are the ones that define the exploration power for three agents; the main difference between the two is how the exploration of the wedge is divided between the agents.

With three agents, we can take advantage of the full exploration power of scheduled meetings. Two agents, \agentB{} and \agentC{} act as \emph{beacons} for an explorer (unbounded agent) \agentA{}. The beacons function is to mark where the unbounded agent turns around and heads in the opposite direction. In the terms we have defined, each meeting of \agentA{ }with \agentB{ }and \agentC{ }is an arranged meeting and \agentA{} interacts with \agentB{} and \agentC{} but these two do not interact. Another possible exploration is when the agents alternate in changing roles from beacon to explorer. When an explorer meets a beacon, it becomes a beacon and the later an explorer in this type of interaction the agents interact pairwise.

If the beacons move in opposite directions (e.g. \agentB{ }moves always west and \agentC{ }always moves east) then the exploration area will be two opposite half-bands. However, if the beacons move in half-bands such that their respective directions have an angle less than $\pi$, the explorer will visit an unbounded number of cells that neither of the beacons will explore. For example, if the beacon \agentB{ } moves to the north and then to the west and waits for \agentA{ }and the beacon \agentC{ }moves to the north and then to the east and waits for \agentA{ }, then \agentB{ }and \agentC{ }would explore the sets $\{(-i,i),(-i,i+1) | i \in \mathbb{N}\},\{(i,i),(i,i+1) | i \in \mathbb{N}\}$ respectively; while \agentA{ }would explore the set $\{(x,i) | i \in \mathbb{N}, -i \leq x \leq i\}$. 

We call the shape resulting from a protocol such that the explorer visits cells the beacons do not a \emph{wedge}. See Figures~\ref{fig:wedge_1},\ref{fig:wedge_2}.

\begin{figure}[h]
	\centering
	\includegraphics[width=0.9\textwidth]{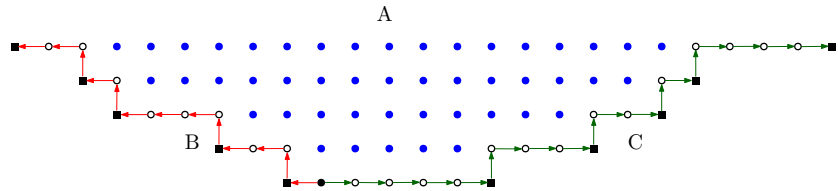}
	\caption{A wedge with the agents starting at the solid black dot. The bounded agents move following the black markers. Each bounded agent waits for the unbounded one at the square cells. The blue dots are the cells explored only by the unbounded agent.}
	\label{fig:wedge_1}
\end{figure}

\begin{figure}[h]
	\centering
	\includegraphics[width=0.4\textwidth]{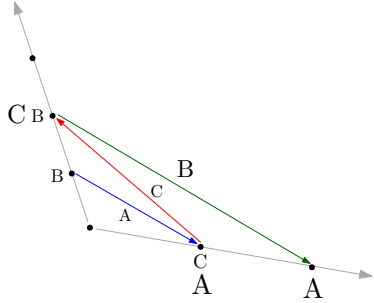}
	\caption{A wedge where the agents exchange roles. The gray vectors are the beacons' directions. \agentA{} takes the role of explorer first, then it changes to beacon and \agentC{} becomes an explorer. When \agentC{} meets \agentB{} it becomes a beacon again and \agentB{} becomes an explorer.}
	\label{fig:wedge_2}
\end{figure}

\begin{theorem}
The exploration power of three agents its a wedge.
\end{theorem}

\begin{proof}
In the previous paragraph we showed the existence of wedges now it remains to show that this type of shapes are the ones with the maximum area. By Corollary~\ref{3exp_pow_1} we know the maximum area three agents can explore in the absence of communication is three half-bands. Corollary~\ref{exp_arranged} implies the exploration area is two half-bands when two agents interact and one does not with either. If \agentA{ }and \agentB{} interact and \agentA{} and \agentC{} also interact but those interactions are arranged, then by Proposition~\ref{exp_arranged} the exploration area is a half-band. Hence the interactions use scheduled meetings. 

Two cases arise, the first one is when \agentA{} only interacts with \agentB{} and \agentC{}; the second when \agentB{} also interacts with \agentA{}. 

If \agentB{} and \agentC{} move in half-bands and \agentA{} bounce between them, If \agentB{}'s and \agentC{}'s directions are parallel (including opposite directions) then the exploration area is a half-band, however if their directions make an angle less than $\pi$ with each iteration \agentB{} and \agentC{} grow apart and the exploration area is a wedge. 
 
On the other hand if \agentB{} and \agentC{} also interact we must show they explore the same wedge. W.l.o.g suppose \agentB{} is a beacon and \agentA{} meets with \agentC{} inside a cell and they exchange places, i.e. \agentA{} becomes a beacon and \agentC{} becomes the explorer (see Figure~\ref{fig:wedge_2}). If \agentC{} does not explore the wedge, it must go around it, however, it requires \agentC{} to turn around at some point to meet \agentB{}; by Corollary~\ref{exp_pow_1ant} it can not do that, hence \agentC{} explores inside the wedge. \qed
\end{proof}

\subsection{Two beacons, one explorer}
\label{2beac}
In this subsection we characterize the type of half-bands three agents can explore. They are characterized based on the location (in terms of plane's quadrants) of the beacon's directions. If both vectors are located in adjacent quadrants (e.g. first and second) the explorer movements between beacons is only following one direction. In Figure~\ref{fig:acute_wedge} we show a wedge in which one direction is $(-6, 1)$ and the other is $(3,7)$. To guide the exploration, the red beacon stops at each cross and the blue beacon stops at each cell, per row, closest to the blue line. The explorer then moves W and E bouncing between the beacons. 

\begin{figure}[h]
	\centering
	\includegraphics[angle=90, width=0.6\textwidth]{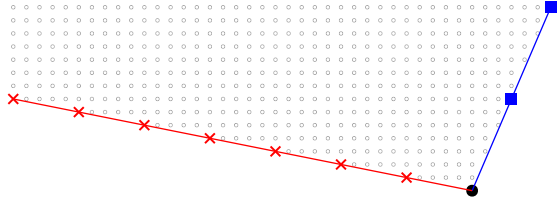}
	\caption{An example of a wedge in which the directions of the beacons are inside adjacent quadrants.}
	\label{fig:acute_wedge}
\end{figure}

In general if a direction is $(i,j)$ then the beacon must visit all the cells of the intersection between the wedge and the rectangle with sides of length $i, j$. 

\begin{figure}[h]
	\centering
	\subfloat[]{
		\includegraphics[width=0.4\textwidth]{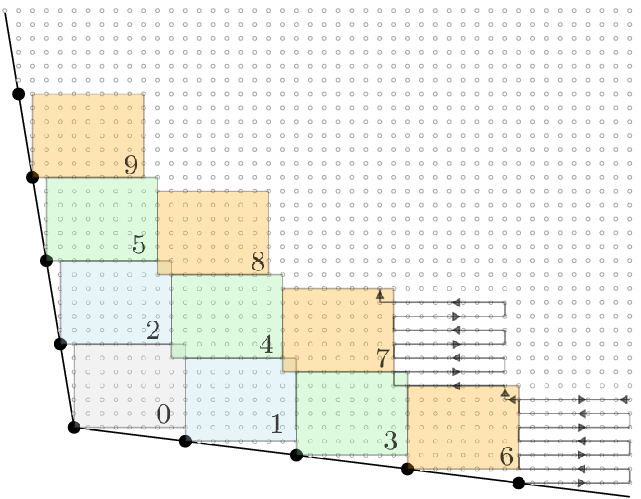}
	}
	\subfloat[]{
		\includegraphics[width=0.43\textwidth]{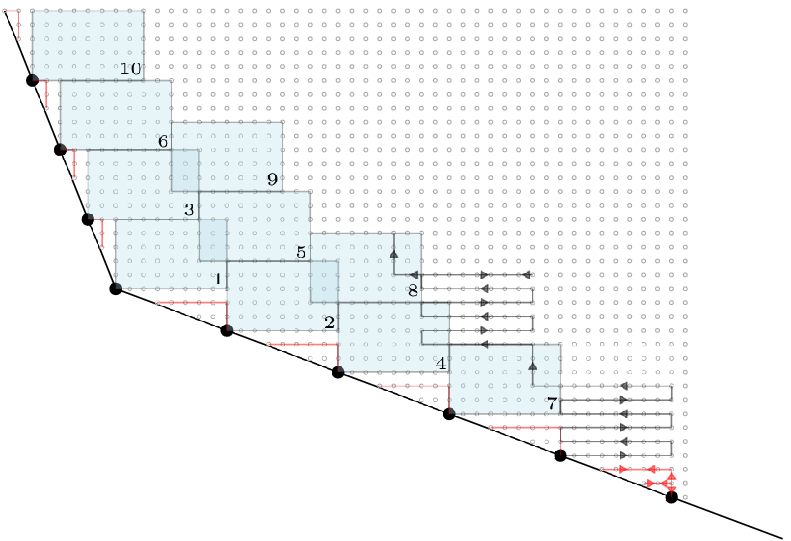}
		\label{fig:5by8}
	}
	\caption{Two examples of an exploration protocol when the directions are not in adjacent quadrants. The number on each rectangle is the order in which \agentA{} explore. The black trajectories are the rectangle exploration \agentA{} performs, the red ones is the extra cells \agentA{} explores when leading \agentB{} and \agentC{} to their next positions.}
	\label{fig:6by8}
\end{figure}

\begin{lemma}
Without loss of generality, let $\mathbf{d}_\agentC = (k, l)$, $\mathbf{d}_\agentB = (i,j)$ such that $i,l > 0$ and $j,k < 0$. There exists a protocol for \agentA{}, \agentB{} and \agentC{}, such that \agentA{} explores the wedge defined by \agentB{} and \agentC{}.
\end{lemma}

\begin{proof}

(Sketch). The proof is by double induction on the absolute value of $k$ and $j$. See Figure~\ref{fig:6by8} for an example of $j = -1, k = -1, i = 8$ and $k = 6$.

When $k = -1$ and $j = -1$ the rectangles \agentA{} explore are of size $i \times l$. \agentC{} movement pattern is $N^lW$. \agentB{} movement pattern is $E^iS$. 

After meeting \agentB{} (or at the start of the protocol), \agentA{} will move $(E^iNW^iN)^{\frac{l}{2}}S$ if $l$ is even, otherwise it will move $(E^iNW^iN)^{\floor*{\frac{l}{2}}}E^iW^i$. Please note that when \agentA{} finishes exploring this first rectangle \agentB{} is already at the next meeting position. Then \agentA{} will move once due west, if \agentA{} does not find \agentC{} it moves once $N$ and moves $(W^iNE^iN)^{\frac{l}{2}}SW^i$ if $l$ is even, otherwise it moves $(W^iNE^iN)^{\floor*{\frac{l}{2}}}W^i$; then moves once west and check if it finds \agentC{} waiting. \agentA{} meets \agentC{} because the path of the later is included in the path of the former.
\qed
\end{proof}

As we will see in the next subsection, it is not necessary that \agentA{} checks for \agentB{} and \agentC{}'s presence when it explores, because any protocol that uses scheduled meetings is equivalent to a protocol in which the unbounded agent \emph{leads} the bounded agent to the next position after a meeting. This allow us an easier analysis of the cases where $|k|, |j| > 1$ because the areas that can not be expressed in rectangles (red paths in Figure~\ref{fig:5by8}) can be explored while the unbounded agent leads the bounded one.

\subsection{A consequence}

Let $\mathbb{P}$ be a protocol such that agents \agentA{} (unbounded agent) and \agentB{} (bounded agent) interact only with scheduled meetings and also suppose that \agentB{} only interacts with \agentA{}. Lemma~\ref{ta_tbCoro} implies that agent \agentB{} can be replaced by a \emph{pebble} that \agentA{} can pick up and place. The difference between a \emph{pebble} and an agent is that the exploration power of the former is just one cell while by Proposition~\ref{1exp_pow} the exploration power of the later is a half-band.

\begin{lemma}
\label{pebble}
Let $\mathbb{P}$ be a protocol such that agents \agentA{} and \agentB{} use only scheduled meetings. Let \agentA{} be the unbounded agent and \agentB{} the bounded one. If \agentB{} does not interact with any other agent, then exists an equivalent protocol $\mathbb{O}$ such that \agentB{} can be replaced by a pebble that \agentA{} can pick up and place.
\end{lemma}

\begin{proof}
To prove the lemma we first construct a protocol $\mathbb{P}'$ in which we remove the need for \agentA{} to expect to meet \agentB{} in any other state but the waiting state. Then we construct $\mathbb{O}$ from $\mathbb{P}'$. In this proof we also assume that there are more than two agents because by Lemma~\ref{2antsSchedMeet} scheduled meetings for only two agents have the same exploration power than one agent.

Let $c_i$ be a cell in which \agentA{} and \agentB{} meet, and $c_k$ the cell of their next meeting. Let $D = (h_1,\ldots,h_m), h_l \in \Delta_\setminus P$ be the directions \agentB{} takes from $c_i$ to $c_k$. Under $\mathbb{P}$ \agentA{} waits for \agentB{} to leave $c_i$ (by means of a waiting cycle $\overline{\omega_\agentB}$) and then follows a chain $L_\agentA$ until it meets back with \agentB{}. Let $L_i^k = q_F^{0} \ldots q_F^{m}\omega_\agentB q_B^{m} \ldots q_B^{0}$ be a chain such that each $q_F^{l} q_F^{l+1}$ transition outputs the direction $h_{l+1}$, $\exit{{\omega_\agentB}} = \omega_\agentA$ and $q_B^{l} q_B^{l+1}$ outputs the direction $-h_{l+1}$. Let $Q'_\agentA = Q_\agentA \cup \{q_F^0, \ldots, q_F^m, q_B^0,\ldots, q_B^m\}$, $\delta'(a,b) = \delta(a,b)\, \forall a \in Q_\agentA, b \in \Sigma_\agentA$, $\delta'(q_F^{l}, \emptyset) = (q_F^{l+1},h_{l+1})$, and for each $\sigma \in \Sigma_\agentA \setminus \{\emptyset, \tau\}, \delta'(q_F^{l}, \sigma) = (q_F^{l}, P)$. Finally $\delta'(q_F^{l}, \tau) = h$ where $h$ is the halting state. Analogously for the states $q_B^l$. Let $\mathbb{P}'$ be the protocol obtained by replacing \agentA{} by \agentAp{} such that $\agentAp = \langle Q'_\agentA, \Sigma_\agentA, \Delta, \delta', q_\agentA^0 \rangle$.

By construction of \agentAp{} after \agentB{} leaves $c_i$, \agentAp{} starts to move towards $c_k$ following $L_i^k$. Once it reaches $c_k$ it enters in a waiting state for \agentB{} ($\omega_\agentB$) and its exit condition is the waiting state of \agentB{} ($\exit{{\omega_\agentB}} = \omega_\agentA$). When \agentAp{} exits its waiting state, \agentB{} is already waiting for \agentAp{} at $c_k$. Hence when \agentAp{} follows $L_\agentA$ there will be no need to expect a meet with \agentB{} in any of it states but the waiting state. The chain $q_B^m \ldots q_B^0$ leads \agentAp{} from $c_k$ back to $c_i$.

If we enrich \agentAp{} with a pebble and the capability to sense it, pick it up and drop it at a cell, we can replace $\overline{\omega_\agentB}$ by $\pick$ and $\omega_\agentB$ by $\drop$. To sense the pebble we add the $\rho$ symbol to $\Sigma_\agentAp$, it is not necessary to add the symbol to the other agents as they may remain oblivious to the pebble's presence. \qed


\end{proof}

Lemma~\ref{pebble} implies the 4 agents protocol in \cite{emek2015many} and our wedge construction where \agentB{} and \agentC{} do not interact and \agentA{} interact with both can be replaced by just one agent and 3 pebbles in the former and one agent and 2 pebbles in the later.

The idea of the previous construction is also valid when $n$ unbounded agents have scheduled meetings with the same bounded agent even if the meetings are at $n$ different locations. The key idea is to set an unbounded agent to act as a \emph{leader} who is responsible to pick up and to drop the pebble. 

Let $\agentA_1, \ldots, \agentA_n$ be $n$ unbounded agents and \agentB{} a bounded agent such that $\agentA_i$ has scheduled meetings with \agentB{}. W.l.o.g. assume $\agentA_i$ meets with \agentB{} before, in time, than $\agentA_{i+1}$. Also asume that $\agentA_1, \ldots, \agentA_n, \agentB$ start at the same cell $c$. $\agentA_1$ is the first agent to have a meet with \agentB{} at cell $c_1^0$ and $\agentA_n$ is the last one to have a meet with \agentB{} at cell $c_n^0$. We argue that the path \agentB{} follows from $c$ to $c_n^0$ is a subpath common to all the unbounded agents. By Lemma~\ref{ta_tbCoro} the cells from $c_i^0$ to $c_i^1$ that \agentB{} visits are also visited by $\agentA_i$ hence the cells from $c_i^k$ to $c_i^{k+1}$ that \agentB{} visits are also visited by all agents $\agentA_j, 1 \leq j \leq n$. If we set $\agentA_1$ as leader, then $\agentA_1$ would move from $c$ to $c_1^0$ and drop the pebble, come back to $c$ and wait for all agents to leave $c$; then move until it reaches $c_1^0$ for the last time and wait for all agents to arrive. Then pick up the pebble and move to $c_2^0$ drop it and go back to $c_1^0$ and so on.

\section{Grid explorations (as conclusion)}
\label{conc}
As a matter of conclusion, we note that when wedges with directions belonging to no adjacent quadrants of the integer plane are considered, then \emph{one explorer, three beacons} type of algorithms only need another agent who's direction will be in a quadrant adjacent to the other two. By the remark at the beginning of Subsection~\ref{2beac}, the movement of the unbounded agent seems to be easy. For example with only an agent \agentD{} that moves $W$ and $S$ and then enters in a waiting state, any wedge with directions $\mathbf{d}_\agentB = (i,-1)$ and $\mathbf{d}_\agentC = (-1,l)$ serves to explore $\mathbb{Z}^2$. However, additional care must be taken for wedges such that $j,k < -1$.

\bibliographystyle{splncs}
\bibliography{biblio.bib}{}

\end{document}